\documentclass{llncs}

\usepackage[lncs,amsthm,autoqed]{pamas}
\usepackage{arydshln}

\usepackage{times}
\usepackage{txfonts}
\usepackage{array,xspace}
\usepackage{multirow,epic}
\usepackage{natbib}
\usepackage{color}
\usepackage{algorithm}
\usepackage{algorithmic}
\usepackage{booktabs}  
\usepackage{verbatim,ifthen}
\usepackage{enumitem}
\usepackage{pifont}
\usepackage{ifthen}
\usepackage{calrsfs,mathrsfs}
\usepackage{cases}

\newcommand\eat[1]{}

\usepackage{enumitem}
\setenumerate[1]{label=\rm(\it{\roman{*}}\rm),ref=({\it\roman{*}}),leftmargin=*}

\newlength{\wordlength}

\newcommand{\eqclass}[2][]{\ifthenelse{\equal{#1}{}}{[#2]}{[#2]_{\sim_{#1}}}}

\newcommand{\Pref}[1][]{
	\ifthenelse{\equal{#1}{}}{\mathrel R}{\mathop{R_{#1}}}
}                                          
\newcommand{\sPref}[1][]{                  
	\ifthenelse{\equal{#1}{}}{\mathrel P}{\mathop{P_{#1}}}
}                                          
\newcommand{\Indiff}[1][]{                 
	\ifthenelse{\equal{#1}{}}{\mathrel I}{\mathop{I_{#1}}}
}
\newcommand{\prefset}[1][]{\ifthenelse{\equal{#1}{}}{\mathcal{R}}{\mathcal{R}_{#1}}}

	\newcommand{\MG}{MG}

\title{Shapley Meets Shapley}

\author{%
	Haris Aziz \and Bart de Keijzer}

\institute{%
NICTA and University of New South Wales, Sydney, NSW 2033, Australia
	\email{haris.aziz@nicta.com.au}
	\and
Centrum Wiskunde \& Informatica (CWI), 1098 XG, Amsterdam, The Netherlands \email{keijzer@cwi.nl}
}

\begin{document}

\maketitle

\begin{abstract}
This paper concerns the analysis of the \emph{Shapley value} in \emph{matching games}.
Matching games constitute a fundamental class of cooperative games which help understand and model auctions and assignments. In a matching game, the value of a coalition of vertices is the weight of the maximum size matching in the subgraph induced by the coalition. The Shapley value is one of the most important solution concepts in cooperative game theory. 

After establishing some general insights, we show that the Shapley value of matching games can be computed in polynomial time for some special cases: graphs with maximum degree two, and graphs that have a small modular decomposition into cliques or cocliques (complete $k$-partite graphs are a notable special case of this). The latter result extends to various other well-known classes of graph-based cooperative games.

We continue by showing that computing the Shapley value of unweighted matching games is $\mathsf{\#P}$-complete in general. Finally, a fully polynomial-time randomized approximation scheme (FPRAS) is presented. This FPRAS can be considered the best positive result conceivable, in view of the $\mathsf{\#P}$-completeness result.
\end{abstract}

\section{Introduction}

In economics and computer science, one of the most fundamental problems is the allocation of profits or costs based on contributions of the nodes in a network. The problem has assumed even more importance 
as networks have become ubiquitous.
In this paper, we address this problem by simultaneously studying two concepts that can be traced to Lloyd S. Shapley --- the \emph{Shapley value} and \emph{matching games}.

%
%



%
%
%
%
%
%
%
%
%
%
%
%
%

Lloyd S. Shapley is one of the most influential game theorists in history. Among his numerous contributions, two of them are the following: (i) formulating the \emph{assignment game} as a rich and versatile class of cooperative games~\citep{ShSh72a}, and (ii) proposing the \emph{Shapley value} as a highly desirable solution concept for cooperative games~\citep{Shap53c}. Both contributions have had far-reaching impact and were part of Shapley's Nobel Prize winning achievements.
The assignment game is a cooperative game based on bipartite graphs, and models the interaction between buyers and sellers. It is the \emph{transferable utility} version of the well-known stable marriage setting and is a fundamental model that is used for modelling exchange markets and auctions ~\citep{RoSo90a}. 
Assignment games were later generalized to \emph{matching games}, for \emph{non}-bipartite graphs~\citep[see \eg][]{DIN99a,KePa03a}. The main idea of a matching game is that each node represents an agent and the value of a coalition of nodes is the weight of the maximum weight matching in the subgraph induced by the coalition of nodes.
Whereas the matching game is one of the most natural and important cooperatives game, the Shapley value has been termed ``the most important normative payoff division scheme'' in cooperative game theory~\citep{Wint02a}. It is based on the idea that the payoff of an agent should be proportional to his marginal contributions to the payoff for the set of all players.
For an excellent overview of the concept, we refer the reader to \citep[Chapter 5, ][]{Moul03a}.
The Shapley value is the only solution concept that satisfies simultaneously the following properties: efficiency, symmetry, additivity, and dummy player property. 

In this paper we address a gap in the computational cooperative game theory literature, and we initiate research on the computational aspects of the Shapley value in matching games. This gap is surprising on two fronts: (i) computational aspects of Shapley values have been extensively studied for a number of cooperative games~\citep[see \eg][]{DePa94a,IeSh05a,EGGW09a}. and (ii) matching games are a well-established class of cooperative games, and the structure and computational complexity of computing important solution concepts such as the core, least core, and nucleolus have been examined in-depth for matching games~\citep[see \eg][]{AlGa90a,SoRa94a,KePa03a,DeFa08a}.




\paragraph{Our results.}
We study the algorithmic aspects and computational complexity of the Shapley value for matching games for the first time. 
We establish first some general insights and some particular special cases for which the exact Shapley value can be computed in polynomial time for: graphs with a constant size decomposition into clique and coclique modules (these include e.g., complete $k$-partite graphs, for $k$ constant), and for graphs with maximum degree two. The non-trivial algorithm required for graphs of maximum degree two illustrates that exact computation of the Shapley value quickly becomes rather complex, even for very simple graph classes.
%
We then move on to the central results of this paper, which concerns the general problem: We prove that the computational complexity of computing the Shapley value of matching games is \#P-complete even if the graph is unweighted. The proof relies on Berge's Lemma and the fact that a certain matrix related to the Pascal triangle has a non-zero determinant.
We subsequently present an \emph{FPRAS} (i.e., a \emph{fully polynomial time randomized approximation scheme}) for computing the Shapley value of (weighted) matching games. In view of our \#P-completeness result, the FPRAS is best possible result we can hope for. 
Due to space limitations, some proofs in this text have been deferred to the appendix.
\paragraph{Related Work.}
The complexity of computing the Shapley value of important classes of cooperative games has been the topic of detailed studies.
\citet{DePa94a} and \citet{IeSh05a} presented polynomial-time algorithms to compute the Shapley value of \emph{graph games} and 
\emph{marginal contribution nets} respectively.
On the other hand, computing the Shapley value is known to be intractable for a number of cooperative games ~\citep[see \eg][]{EGGW09a, ALPS09b}.  

%
Among the classes of cooperative games, matching games are one of the most well-studied. \citet{DIN99a} characterized the core of the matching games and showed that various problems regarding the core and the least core of matching games can be solved in polynomial time. 
For matching games, there has been considerable algorithmic research on the \emph{nucleolus}: an alternative single valued solution concept\citep[see \eg][]{SoRa94a,KePa03a}.

The Shapley value of a vertex in a matching game indicates the ability of a vertex to match with other vertices. It may thus also be viewed as a centrality index of a vertex. Centrality indices of graphs have received immense interest~\citep[see \eg][]{BrEr05a}.

\section{Preliminaries}\label{sec:preliminaries}

We work throughout this text with undirected weighted graphs $G = (N,E,w)$, where $N$ is the vertex set, $E$ is the edge set, and $w : E \rightarrow \mathbb{R}_{\geq 0}$ is a weight function.
For $S \subseteq N$, we denote by $G(S)$ the subgraph of $G$ induced by $S$, i.e., the graph $(S, \{e \in E : e \in S \times S\})$. Some essential basic notions related to graphs and matchings may be found in the appendix. We assume for the remainder of this text that the reader is familiar with these.

%

A \emph{cooperative game} consists of a set $N$ of $n = |N|$ players and a characteristic function $v : 2^N\rightarrow \mathbb{R}$ associating a value $v(S)$ to every subset $S\subseteq N$. A subset of $N$ is referred to as a \emph{coalition} in this context. A central question in the theory of cooperative games is to distribute the value $v(N)$ among the players in a fair and stable manner.

A \emph{matching game} is a cooperative game $(N,v)$ induced by an undirected weighted graph $G = (N,E,w)$ (with vertex set $N$, edge set $E$, and weight function $w: E\rightarrow \mathbb{R}_{\geq 0}$) such that for any $S \subseteq N$, $v(S)$ is the weight of a maximum weight matching of the subgraph $G(S)$. For a given graph $G$, we will denote by $\MG(G)$ the matching game corresponding to graph $G$. 

An \emph{unweighted} matching game is a matching game for which all weights are $1$ in the associated graph. In unweighted matching games, it holds that $v(S \cup \{i\}) - v(S) \in \{0,1\}$ for all  $S \subset N$, $i \in N \backslash S$. 
If, for an unweighted matching game $(N,v)$, a player $i \in N$, and a coalition $S \subseteq N \backslash \{i\}$, it holds that $v(S \cup \{i\}) = v(S)+1$, then we say that player $i$ is \emph{pivotal} (for coalition $S$, in game $(N,v)$). Similarly, if $\sigma : N \rightarrow N$ is a permutation on $N$, and $i$ is pivotal for set of players $p(i,\sigma) = \{j : \sigma^{-1}(j) < \sigma^{-1}(i)\}$ (i.e., the players occurring before $i$ in $\sigma$) is pivotal, then we say that $\sigma$ is pivotal for $i$. 

For the general case of weighted matching games, when $S$ is a coalition not containing player $i$, we refer to the value $v(S \cup \{i\}) - v(S)$ as \emph{the marginal contribution of $i$ to $S$.} When $\sigma$ is a permutation of $N$, we refer to the value $v(p(i,\sigma) \cup \{i\}) - v(p(i,\sigma))$ as \emph{the marginal contribution of $i$ to $\sigma$.}

The \emph{Shapley value} of a player $i\in N$ in a cooperative game $(N,v)$ is denoted by $\varphi_i(N,v)$, and is defined as follows.
\begin{equation}\label{eq:rawshapley1}
\varphi_i(N,v)=\kappa_i(N,v)/{|N|!}, \qquad \kappa_i(N,v)=\sum_{S\subseteq N\setminus \{i\}}(|S|!)(|N|-|S|-1)!(v(S\cup \{i\})-v(S)).
\end{equation}
$\kappa_i$ is called the \emph{raw Shapley value}. It is well-known and straightforward to obtain that the raw Shapley value can be written as
$\kappa_i(N,v) = \sum_{\sigma \in S_N} (v(p(i,\sigma) \cup \{i\}) - v(p(i,\sigma),)),$
where $S_N$ is the set of permutations on the player set $N$.
For an unweighted matching game, the raw Shapley value of a player is thus equal to the number of pivotal permutations.
We refer to the vectors $\varphi = (\varphi_1(N,v), \ldots \varphi_n(N,v))$ and $\kappa = (\kappa_1(N,v), \ldots, \kappa_n(N,v))$ respectively as the Shapley value and the raw Shapley value of the game \emph{$(N,v)$}.



The players $i,j\in N$ are called \emph{symmetric} in $(N,v)$ if $v(S\cup \{i\})=v(S\cup\{j\})$ for any coalition $S\subseteq N\setminus \{i,j\}$. A player $i\in N$ is a \emph{dummy} if $v(S\cup\{i\})-v(S)=0$ for all $S\subseteq N$. The Shapley value satisfies the following properties: 
(i) \emph{Efficiency}: $\sum_{i\in N}\varphi_i(N,v)= v(N)$;
(ii) \emph{Symmetry}: if $i,j\in N$ are symmetric, then $\varphi_i(N,v)=\varphi_j(N,v)$;
(iii) \emph{Dummy}: if $i$ is a dummy, then $\varphi_i(N,v)=0$;
(iv) \emph{Additivity}: $\varphi_i(N,v^1+v^2)=\varphi_i(N,v_1) + \varphi_i(N,v_2)$ for all $i\in N$;\footnote{The sum of two characteristic functions $v_1$ and $v_2$ on the same player set is defined in the standard way: as $v_1(S)+v_2(S)$ for all $S \subseteq N$.} and
(v) \emph{Anonymity}: relabeling the agents does not affect their Shapley value.
We are interested in the following computational problem.\\\\
\noindent
{\sc Shapley}  \\
\noindent
Instance: A weighted graph $G=(N,E,w)$ and a specified player $i\in V$\\
\noindent
Question: Compute $\varphi_i(\MG(G)).$

\subsection{General insights}
In this subsection, we gain some general insights about the Shapley value of matching games. 
First, if the graph is not connected, then the problem of computing the Shapley value of the graph reduces to computing the Shapley value of the respective connected components.
\begin{lemma}[Shapley value in connected components]\label{lemma:components}
Let $G = (N,E,w)$ be a weighted graph with $k$ connected components, and let the respective vertex sets of these connected components be $N_1, \ldots, N_k$. Let $v$ be the characteristic function of the matching game $\MG(G)$ on that graph, and let $c : N \rightarrow [k]$ be the function that maps a vertex $i$ to the number $k$ such that $i \in N_k$.\footnote{For $a \in \mathbb{N}$, we write $[a]$ to denote $\{b \in \mathbb{N} : 1 \leq b \leq a\}$.} Then, for every vertex $i$ it holds that $\varphi_i(v) = \varphi_i(v_{c(i)})$, where $v_j$ denotes the characteristic function of the matching game on the subgraph induced by $V_j$.
\end{lemma}
It is rather straightforward to see that a vertex has a Shapley value zero if and only if it is not connected to any other vertex.
\begin{observation}\label{zeroshapley}
A player in a matching game has a non-zero Shapley value if and only if there is an edge in the graph that contains the player. It can thus be decided in linear time whether a player in a matching game has a Shapley value of zero.
\end{observation}

Next, we present another lemma concerning the Shapley value of unweighted matching games.
\begin{lemma}\label{lemma:basic}
Consider an unweighted matching game $(N,v)$. 
If for each $s \in [n-1]$, the number of coalitions of size $s$ for which player $i$ is pivotal in $(N,v)$ can be computed in time $f(n)$ for some function $f : \mathbb{N} \rightarrow \mathbb{R}_{\geq 0}$, then the Shapley value of $i$ can be computed in time $n f(n)$.
\end{lemma}

\section{Exact algorithms for restricted graph classes}
Some classes of matching games for which computing the Shapley value is trivial are symmetric graphs (e.g. cliques and cycles), and graphs with a constant number of vertices.
We proceed to prove this for two additional special cases: graphs that admit constant size (co)clique modular decompositions, and graphs with degree at most two.

\vspace{-0.5em}

\subsection{Graphs with a constant number of clique or coclique modules}
An important concept in the context of undirected graphs is that of a \emph{module}. A subset of vertices $S \subseteq N$ is a module if all members of $S$ have the same set of neighbors in $N \setminus S$. We can extend this notion to weighted graphs by requiring that all members of $S$ are connected to the same set of neighbors, by edges of the same weight. A \emph{modular decomposition} is a partition of the vertex set into modules.

A \emph{clique module} (resp. \emph{coclique module}) of a weighted graph is a module of which the vertices are pairwise connected by edges of the same weight (resp. pairwise disconnected).
Note that every graph has a trivial modular decomposition into cliques (and cocliques): the partition of $N$ into singletons.

We prove that if an unweighted graph $G$ has a size $k$ modular decomposition consisting of only cliques or only cocliques, then the Shapley value of $\MG(G)$ can be found in polynomial time. In fact, we will show that this holds for the more general class of \emph{subgraph-based} games: We call a cooperative game $(N,v)$ \emph{subgraph-based} if there exists a weighted graph $G = (N,E,w)$ such that for $S, T \subset N$, it holds that $v(S) = v(T)$ if $G(S)$ and $G(T)$ are isomorphic.



\begin{theorem}\label{th:modules}
Consider a subgraph-based cooperative game $(N,v)$. Then, the Shapley value of $(N,v)$ can be computed in polynomial time if the following conditions hold: i.) the weighted graph $G = (N,E,w)$ associated to $(N,v)$ is given or can be computed from $(N,v)$ in polynomial time; 2.) there exists a modular decomposition $\gamma(G)$ into $k$ cocliques or $k$ cliques and $G$ is unweighted in the latter case; and iii.) $v(S)$ can be computed in polynomial time for all $S \subseteq N$. 
\end{theorem}
\begin{proof}
Note first that one can find for $G$ in polynomial time a minimum cardinality modular decomposition into cocliques: simply check for each pair of vertices whether they are disconnected and connected to identical sets of vertices through edges with identical weights. If so, then they can be put in the same module. Similarly, a minimum cardinality modular decomposition into cliques can be found in polynomial time in case the graph is unweighted, by finding a minimum cardinality modular decomposition into cocliques in the complement of $G$ (i.e., the graph that contains only those edges not in $E$). 

A set of players $S$ is said to be of the same \emph{player type} if all players in $S$ are pairwise symmetric. We first show that all players in the same module of $\gamma(G)$ are of the same player type. Let $i,j$ be two players in the same module $M$ in $\gamma(G)$. Then, for every coalition $C \in N \backslash \{i,j\}$, the subgraphs $G(C \cup \{i\})$ and $G(C \cup \{j\})$ are isomorphic (because $G(M)$ is a clique or coclique), so $v(C \cup \{i\}) = v(C \cup \{j\})$. Therefore, we know that the vertices can be divided into a constant number $k$ of player types.
	
\citet{UKI+11a} showed that any cooperative game in which the value of a given coalition can be computed in polynomial time, and there is known size $k$ partition of the players into sets of the same player type, then the Shapley value can be computed in polynomial time via dynamic programming. The number of player types in our game is constant number $k$ of clique and coclique modules, and therefore the result of \cite{UKI+11a} can be applied, and proves our claim.
\end{proof}

For matching games, the function $v$ can be evaluated using any polynomial time maximum weight matching algorithm. Therefore, the above result implies that computing the Shapley value can be done in polynomial time for classes of graphs where we can find efficiently a size $k$ modular decomposition into cliques or cocliques. This includes the class of complete $k$-partite graphs and any strong product\footnote{The \emph{strong product} of two graphs $G_1 = (N, E_1)$ and $G_2 = (M, E_2)$ is defined as the graph $(N \times M, E')$, where $E' = \{\{(i_N,i_M), (j_N, j_M)\} \subseteq N \times M : i_M = j_M \wedge \{i_N,j_N\} \in E_1 \vee \{i_M, j_M\} \in E_2\}$.} of an arbitrary size clique (or coclique) with a graph on $k$ vertices. 

\begin{corollary}\label{th:kpartite}
For matching games based on complete $k$-partite graphs, where $k$ is a constant, the Shapley value can be computed in polynomial time. 
\end{corollary}

Theorem~\ref{th:modules} also applies to cooperative games such as $s$-$t$ vertex connectivity games and min-cost spanning tree games~\citep{DeFa08a,DIN99a}, as these are subgraph-based games.

\subsection{Graphs of degree at most two}
We first examine \emph{linear graphs} (or: ``paths''), i.e., connected graphs in which two vertices have out-degree one and the remaining vertices have out-degree two.

\begin{lemma}\label{lemma:line}
The Shapley value of a player in a matching game on an unweighted linear graph can be computed in $O(n^4)$ time.
\end{lemma}
\begin{proof} 
Assume without loss of generality that the vertex set is $[n]$ and the edge set is $\{\{j, j+1\} : j \in [n-1]\}$, and that $i \in [n]$ is the player of whom we want to compute the Shapley value. Fix any $s \in [n-1]$, and let $\eta_i^s$ be the number of coalitions of size $s$ for which vertex $i$ is pivotal.
We compute $\eta_i^s$ by subdividing in separate cases and taking the sum of them:
\begin{itemize}
 \item The number $\eta_i^{s, \text{left}} = |\{S \cup \{i + 1\} : S \in N \backslash \{i,i-1,i+1\}, i \text{ is pivotal for } S \}|$. Intuitively: the number of coalitions $S$ where $i$ is pivotal such that adding $i$ to $S$ extends the left of a line segment.
 \item The number $\eta_i^{s, \text{right}} = |\{S \cup \{i - 1\} : S \in N \backslash \{i,i-1,i+1\}, i \text{ is pivotal for } S \}|$.
 \item The number $\eta_i^{s,\text{connect}} = |\{S \cup \{i - 1, i + 1\} : S \in N \backslash \{i,i-1,i+1\}, i \text{ is pivotal for } S \}|$. Intuitively: the number of coalitions $S$ where $i$ is pivotal, such that $i$ connects two line segments.
 \item $\eta_i^{s,\text{isolated}} = |\{S \backslash \{i - 1, i + 1\} : S \in N \backslash \{i,i-1,i+1\}, i \text{ is pivotal for } S \}|$.
\end{itemize}

It is immediate that $\eta_i^{s,\text{isolated}} = 0$, since adding $i$ to a coalition $S$ not containing $i+1$ nor $i-1$ results in a coalition forming a subgraph in which $i$ is an isolated vertex. 
For the remaining three values, $\eta_i^{s,\text{left}}$, $\eta_i^{s,\text{right}}$, and $\eta_i^{s,\text{connect}}$, we show below how to compute them efficiently.
\begin{itemize}
 \item For $\eta_i^{s,\text{left}}$, observe that adding a vertex to the left of a (non-empty) line segment $L$ increases the cardinality of a maximum matching if and only if $L$ has an even number of edges (and thus an odd number of vertices). Therefore, define $\eta_i^{s,\text{left}}(k)$ to be the number of coalitions $S$ of size $s$ for which $i$ is pivotal such that $S$ contains the line segment $\{i+1, \ldots, i+k+1\}$, and does not contain $\{i-1, i+k+2\}$. The number $\eta_i^{s,\text{left}}(k)$ is easy to determine:
\begin{equation*}
\eta_i^{s,\text{left}}(k) = \begin{cases} 0 & \text  { if } k \text{ is odd, } \\ \binom{|[n] \backslash \{i - 1, \ldots, i+k+2\}|}{s - |\{i - 1, \ldots, i+k+1\} \cap [n]|} & \text{ otherwise. } \end{cases}
\end{equation*}
We can then express $\eta_i^{s,\text{left}}$ as $\sum_{k=1}^{\max\{n-i-1, s-1\}} \eta_i^{s,\text{left}}(k)$. There is only a linear number of terms in this sum, and all of them can be computed in linear time. 
 \item $\eta_i^{s,\text{right}}$ is computed in an analogous fashion.
 \item For $\eta_i^{s,\text{connect}}$, observe that adding a vertex $i$ to a coalition such that $i$ connects two line segments $L_1$ and $L_2$, increases the cardinality of a maximum matching if and only if $L_1$ and $L_2$ do not both have an odd number of edges (or equivalently: not both have an even number of vertices). Therefore, define $\eta_i^{s,\text{connect}}(k_1, k_2)$ to be the number of coalitions $S$ of size $s$ for which $i$ is pivotal such that $S$ contains the line segments $\{i-k_1-1, \ldots, i-1\}$ and $\{i+1, \ldots, i+k_2+1\}$, and does not contain $\{i-k_1-2, i+k_2+2\}$. The number $\eta_i^{s,\text{connect}}(k_1, k_2)$ is easy to determine:
\begin{equation*}
\eta_i^{s,\text{connect}}(k_1, k_2) = \begin{cases} 0 & \text{ if } k_1 \text{ and } k_2 \text{ are both odd, } \\ \binom{|[n] \backslash (\{i - k - 2, \ldots, i+k+2\}|}{s - |\{i - k - 1, \ldots, i+k+1\} \cap [n]|} & \text{ otherwise. } \end{cases}
\end{equation*}
We can then express $\eta_i^{s,\text{connect}}$ as 
$\sum_{k_1 = 1}^{\max\{i-2, s-1\}} \sum_{k_2 = 1}^{\max\{n-i-1, s-k_1-2\}} \eta_i^{s,\text{left}}(k_1,k_2).$ The number of terms in this sum is quadratic, and all of these terms can be computed in linear time. We can thus compute $\eta_i^{s,\text{connect}}$ in $O(n^3)$ time.
\end{itemize}
The claim now follows from Lemma \ref{lemma:basic}.
\end{proof}

\begin{theorem}
For graphs with maximum degree $2$, the Shapley value can be computed in polynomial time.
\end{theorem}
\begin{proof}
A graph with degree at most two is a disjoint union of cycles and linear graphs. From Lemma~\ref{lemma:components}, we can compute the Shapley value of the connected components separately. From Lemma~\ref{lemma:line}, we know that the Shapley value of linear graphs can be computed in polynomial time. Due to anonymity, the Shapley value of a cycle is uniform.
\end{proof}

The above proof for linear graphs demonstrates nicely that computation of the Shapley value of a matching game already becomes intricate for even the simplest of graph structures. We would be interested in seeing an extension of this result that enables us to exactly compute the Shapley value in \emph{any} non-trivial class of graphs that contains a vertex of degree at least three.

\section{Computational complexity of the general problem}
In this section, we examine the computational complexity of the general problem of computing the Shapley value for matching games. 
As we mentioned in Section \ref{sec:preliminaries}, {\sc Shapley} is equivalent to the problem of counting the number of pivotal permutations for a player in an unweighted matching game, and is therefore a counting problem. It is moreover easy to see that this counting problem is a member of the complexity class $\mathsf{\#P}$.\footnote{Informally: $\mathsf{\#P}$ is the class of computational problems that correspond to counting the number of accepting paths on a non-deterministic Turing machine. We refer the reader to any introductory text on complexity theory.}

For certain cooperative games such as weighted voting games~\citep{EGGW09a}, 
intractability of computing the Shapley value can be established by proving that even checking whether a player gets non-zero Shapley value is $\mathsf{NP}$-complete. Proposition \ref{zeroshapley} tells us that this is not the case for matching games.
Before we proceed, we establish some notation. Let $G=(N,E)$ be a graph. Let $\alpha_k(G)$ be the number of vertex sets $S \subseteq N$ such that $|S|=k$ and the subgraph $G(S)$ of $G$ induced by $S$ admits a perfect matching. Then $\overline{\alpha_k}(G)={n\choose k}-\alpha_k(G)$ is the number of subsets $S \subseteq N$ of size $k$ such that $G(S)$ does not admit a perfect matching. In order to characterize the complexity of {\sc Shapley}, we first define the following problem. \\
\\
%
\noindent
{\sc \#MatchableSubgraphs$_k$}  \\
\noindent
Instance: Undirected and unweighted graph $G=(N,E)$ and an even integer $k$. \\
\noindent
Question: Compute $\alpha_k(G)$.

%

\begin{lemma}\label{lemma:countinghard}
{\sc \#MatchableSubgraphs$_k$} is $\mathsf{\#P}$-complete. 
\end{lemma}
\begin{proof}
\citet{CPV95a} proved that the following problem is $\mathsf{\#P}$-complete: Given an undirected and unweighted bipartite graph $G = (S\cup I,E)$, compute the number of subsets of $B \subseteq  S$, such that $G(B \cup I)$ admits a perfect matching.\footnote{The proof of Colbourn resolved ``an exceptionally difficult problem''~\citep{CPV95a}. Interestingly, the corresponding decision problem of checking whether there exists a subgraph of size $k$ that does not admit a perfect matching, appears to be open.} 
The problem is equivalent to
\#MatchableSubgraphs$_{2|I|}$.
\end{proof}
%
\begin{theorem}
Computing the Shapley value of a matching game on an unweighted graph is $\mathsf{\#P}$-complete. 
\end{theorem}
\begin{proof}
We present a polynomial-time Turing reduction from {\sc \#MatchableSubgraphs$_k$} to {\sc Shapley}.

Let $G_0$ be the graph in which a new vertex $y_0$ is added to $G = (N,E)$ that is connected to all vertices in $N$. For $i>0$, let $G_i$ be $G_0$ with $i$ additional vertices $y_1,y_2, \ldots, y_i$ and $i$ additional edges $\{\{y_j, y_{j-1}\} : j \in [i]\}$. 

The first part of the proof consists of showing that the following set of equations hold:
\begin{numcases}{\kappa_{y_i}(\MG(G_i)) = }
\textstyle{C(i) + \sum_{k=0}^n (k+i)!(n-k)! \overline{\alpha_k}(G)} & \text{ if } i \text{ is even,} \label{eq:reductioneven}\\
\textstyle{C(i) + \sum_{k=0}^n (k+i)!(n-k)! \alpha_k(G)} & \text{ if } i \text{ is odd,} \label{eq:reductionodd}
\end{numcases}
where 
\begin{equation*}C(i) = \sum_{k = 1}^{\lfloor i/2 \rfloor} \sum_{j = 0}^{n+i-2k} (j+2k-1)! (n+i-j-2k+1)! \binom{n+i-2k}{j}.
\end{equation*}

Define a \emph{type 1 pivotal coalition for $y_i$ in $\MG(G_i)$} as a pivotal coalition for $i$ in $\MG(G_i)$ that \emph{does not} contain all players $y_0, \ldots, y_{i-1}$.
Define a \emph{type 2 pivotal coalition for $y_i$ in $\MG(G_i)$} as a pivotal coalition for $y_i$ in $\MG(G_i)$ that \emph{does} contain all players $y_0, \ldots, y_{i-1}$. Denote by $H_i^{\text{type 1}}(s)$ (resp. $H_i^{\text{type 2}}(s)$) the set of type 1 (resp. type 2) pivotal coalitions for $i$ in $\MG(G_i)$ that are of size $s$. From (\ref{eq:rawshapley1}), it follows that
\begin{equation}\label{eq:rawshapleytypes}
\kappa_i(\MG(G_i)) = \sum_{s = 1}^{n+i} s!(n+i-s)!|H_i^{\text{type 1}}(s)| + \sum_{s = 1}^{n+i} s!(n+i-s)!|H_i^{\text{type 2}}(s)|.
\end{equation}

First we characterize the coalitions in $H_i^{\text{type 2}}(s)$.
\begin{lemma}\label{lem:char}
If $i$ is even, a coalition $S$ of $\MG(G_i)$ is in $H_i^{\text{type 2}}(s)$ if and only if $G(S \cap N)$ is not perfectly matchable (and $\{y_0, \ldots, y_{i-1}\} \subseteq S, |S| = s$).
If $i$ is odd, a coalition $S$ of $\MG(G_i)$ is in $H_i^{\text{type 2}}(s)$ if and only if $G(S \cap N)$ is perfectly matchable (and $\{y_0, \ldots, y_{i-1}\} \subseteq S, |S| = s$).
\end{lemma}
%
%
The proof of Lemma \ref{lem:char} is deferred to the appendix.
From the above lemma, it follows that the coalitions in $H_i^{\text{type 2}}(s)$ are precisely the coalitions of the form $T \cup \{y_0, \ldots, y_{i-1}\}$, where $T \subset N$ is such that for even $i$, $G(T)$ is not perfectly matchable, and for odd $i$, $G(T)$ is perfectly matchable. Therefore $|H_i^{\text{type 2}}(s)| = \overline{\alpha_{s - i}}(G)$ for even $i$ and $|H_i^{\text{type 2}}(s)| = \alpha_{s - i}(G)$ for odd $i$, and this implies:
\begin{equation*}
\sum_{s = 1}^{n+i} s!(n+i-s)!|H_i^{\text{type 2}}(s)| = 
\begin{cases}
\sum_{k=0}^n (k+i)!(n-k)! \overline{\alpha_k}(G) & \text{ if } i \text{ is even, } \\
\sum_{k=0}^n (k+i)!(n-k)! \alpha_k(G) & \text{ if } i \text{ is odd. } \\
\end{cases}
\end{equation*}
In words: the second summation of (\ref{eq:rawshapleytypes}) equals the summation of (\ref{eq:reductioneven}) when $i$ is even, and the summation of (\ref{eq:reductionodd}) when $i$ is odd.
Therefore, it suffices to prove that the first summation of (\ref{eq:rawshapleytypes}) equals $C(i)$. 

For this sake, define $H_i^{\text{type 1}}(s,k)$ for $k \in [\lfloor i/2 \rfloor]$ as $\{S \in H_i^{\text{type 1}}(s) : y_{i-2k} \not \in S \wedge \{y_{i-1}, \ldots, y_{i - 2k + 1}\} \subseteq S\}$. Observe that $\{H_i^{\text{type 1}}(s,1), \ldots, H_i^{\text{type 1}}(s,k/2)\}$ is a partition of $H_i^{\text{type 1}}(s)$.
For a given $k$ and $s$, note that the set $H_i^{\text{type 1}}(s,k)$ consists of all coalitions of the form $T \cup \{y_{i-1}, \ldots, y_{i - 2k + 1}\}$, where $T \subseteq N \cup \{y_0, \ldots, y_{i - 2k - 1}\}$, $|T| = s - 2k + 1$. Hence, $|H_i^{\text{type 1}}(s,k)| = \binom{n+i-2k}{s - 2k + 1}$ (defining $\binom{a}{b} = 0$ whenever $b<0$ or $b>a$). Therefore:
\begin{eqnarray*}
\sum_{s = 1}^{n+i} s!(n+i-s)!|H_i^{\text{type 1}}(s)| & = & \sum_{k = 1}^{\lfloor i/2 \rfloor} \sum_{s=2k-1}^{n+i-1} s!(n+i-s)!\binom{n+i-2k}{s-2k+1} \\
& = & \sum_{k = 1}^{\lfloor i/2 \rfloor} \sum_{j=0}^{n+i-2k} (j+2k-1)!(n+i-j-2k+1)!\binom{n+i-2k}{j}.
\end{eqnarray*}
This shows that (\ref{eq:reductioneven}) and (\ref{eq:reductionodd}) hold.

The second part of the proof consists of showing that all $\alpha_k(G), k \in [n]$ can be computed from $\kappa_{y_i}(MG(G_i))$ in polynomial time, using (\ref{eq:reductioneven}) and (\ref{eq:reductionodd}), for $i \in [n] \cup \{0\}$. This is sufficient to complete the proof, because the graphs $G_0, \ldots, G_n$ can clearly be constructed from $G$ in polynomial time, hence a polynomial time algorithm that computes $\alpha_k$ from $\kappa_{y_i}(MG(G_i)), i \in [n]$ is a polynomial Turing reduction.

Let $\beta_i(G) = \alpha_i(G)$ for even $i$ and let $\beta_i(G) = \overline{\alpha_i}(G)$ for odd $i$.
We can represent (\ref{eq:reductioneven}) and (\ref{eq:reductionodd}) for $i \in [n] \cup \{0\}$ as the following system of equations:
\small
\begin{equation}\label{eq:system}
\begin{pmatrix}
 0!n!& 1!(n-1)! & \cdots & n!0!  \\
 1!n!&  & \cdots & (n+1)!0!  \\
\vdots  & \vdots  & \ddots & \vdots  \\
 n!n!&  & \cdots & (2n)!0!  \\ 
\end{pmatrix}\times \begin{pmatrix}
 \beta_0(G)\\
 \beta_1(G)\\
\vdots  \\
\beta_n(G)\\
 \end{pmatrix} =
\begin{pmatrix}
 \kappa_{y_0}(\MG(G_0))-C(0)\\
\kappa_{y_1}(\MG(G_1))-C(1)\\
\vdots  \\
\kappa_{y_n}(\MG(G_n))-C(n)\\
 \end{pmatrix} 
\end{equation}
\normalsize

Denote by $A$ the $(n+1) \times (n+1)$ matrix in the above equation.
Recall that a scalar multiplication of a column by a constant $c$ multiplies the determinant by $c$.
Therefore, $A$ is nonsingular if and only if nonsingularity also holds for the $(n+1) \times (n+1)$ matrix $B$, defined by $B_{ij} = (i+j)!$. 
$B$ is a matrix that is related to Pascal's triangle, and it is known that its determinant is equal to $\prod_{i=0}^n {i!}^2 \neq 0$~\citep{Bach02a,ALPS09b}. It follows that $A$ is nonsingular, so our system of equations (\ref{eq:system}) is linearly independent and has a unique solution. Note that all entries in the system can be computed in polynomial time (assuming that the Shapley value of a matching game is polynomial time computable): The constants $C(i)$ consist of polynomially many terms, and all factorials and binomial coefficients that occur in ($\ref{eq:system}$) are taken over numbers of magnitude polynomial in $n$. 

Therefore, we can use Gaussian elimination to solve (\ref{eq:system}) in $O(n^3)$ time. 
 It follows that for all $i \in [n]$, $\beta_i(G)$ can be computed in polynomial time, and hence $\alpha_i(G)$ can be computed in polynomial time. Therefore, if there exists an algorithm that solves {\sc Shapley} in polynomial time, then it can also be used to solve {\sc \#MatchableSubgraphs$_k$} in polynomial time. 
\end{proof}



%
%
%

\section{An approximation algorithm}

In this section, we show that although computing exactly the Shapley value of matching games is a hard problem, approximating it is much easier.

Let $\Sigma$ be a finite alphabet in which we agree to describe our problem instances and solutions. 
A \emph{fully polynomial time randomized approximation scheme (FPRAS)} for a function $f : \Sigma^* \rightarrow \mathbb{Q}$ is an algorithm that takes input $x \in \Sigma^*$ and a parameter $\epsilon \in \mathbb{Q}_{> 0}$, and returns with probability at least $\frac34$ a number in between $f(x)/(1+\epsilon)$ and $(1+\epsilon) f(x)$. Moreover, an FPRAS is required to run in time polynomial in the size of $x$ and $1/\epsilon$. The probability of $\frac34$ is chosen arbitrarily: by a standard amplification technique, it can be replaced by an arbitrary number $\delta \in [0,1]$. The resulting algorithm would then run in time polynomial in $n, 1/\epsilon$, and $\log(1/\delta)$. 

We will now formulate an algorithm that approximates the raw Shapley value of a player in a weighted matching game, and show that it is an FPRAS. Note that we cannot utilize approximation results in 
\citep{NSW+11a} and \citep{BMR+10a} since matching games are neither convex nor simple.
Our FPRAS is based on Monte Carlo sampling, and works as follows: Let $(G = (N,E,w), i, \epsilon)$ be the input, where $G$ is the weighted graph representing matching game $\MG(G)$, $i \in N$ is a player in $\MG(G)$, and $\epsilon$ is the precision parameter. For notational convenience, we write $\kappa_i$ as a shorthand for $\kappa_i(\MG(G))$. The algorithm first determines whether $\kappa_i = 0$ (Observation \ref{zeroshapley}). If so, then it outputs $0$ and terminates. If not, then it samples $\lceil 4n^2(n-1)^2/\epsilon^2 \rceil$ permutations of the player set uniformly at random. Denote this multiset of sampled permutations by $P$. The algorithm then outputs the average marginal contribution of player $i$ over the permutations in $P$ and terminates. Note that this average marginal contribution is efficiently computable: it is given by $1/\lceil 4n^2(n-1)^2 / \epsilon^2 \rceil$ times the sum of the marginal contributions of player $i$ to each of the sampled permutations. Determining these marginal 
contributions can be done in 
polynomial time, using any maximum weight matching algorithm. Denote our sampling algorithm by \emph{\textsc{MatchingGame-Sampler}}.

\textsc{MatchingGame-Sampler} resembles the algorithms in \citep{MaSh60a,NSW+11a}: the differences are that the algorithm takes a different number of samples, and that it determines whether the Shapley value of player $i$ is $0$ prior to running the sampling procedure. Moreover, its proof of correctness requires different insights.\footnote{To be precise, this applies only to \citep{NSW+11a}. For the sampling algorithm in \cite{MaSh60a}, no proof or approximation-quality analysis of any kind is given.}

\begin{theorem}
\textsc{MatchingGame-Sampler} is an FPRAS for the raw Shapley value in a weighted matching game.
\end{theorem}
\begin{proof}
Denote by $\bar{\kappa}_i$ the output of the algorithm. If $\kappa_i = 0$, then \textsc{MatchingGame-Sampler} is guaranteed to output the right solution, so assume that $\kappa_i > 0$. Let $w_i^{\max}$ be the maximum weight among the edges attached to $i$, and let $e_i^{\max} \in E$ be an edge that is attached to $i$ such that $w(e_i^{\max}) = w_i^{\max}$. Let $X$ be a random variable that takes the value of $n!$ times the marginal contribution of player $i$ in a uniformly randomly sampled permutation of the players. Note that $\mathbf{E}[X] = \kappa_i$. Note that the marginal contribution of a player in any permutation is at most $w_i^{\max}$, so $X$ is at most $w_i^{\max}n!$.

Let $j$ be the neighbor of $i$ connected by $e_i^{\max}$. Observe that any permutation in which $j$ is positioned first, and $i$ is positioned second, is a permutation for $i$ in which the marginal contribution of $i$ is $w_i^{\max}$. There are $(n-2)!$ such permutations, so the raw Shapley value $\kappa_i$ of $i$ is at least $w_i^{\max}(n-2)!$. For the variance of $X$ we obtain 
\begin{align*}
\mathbf{Var}[X] & = \mathbf{E}[X^2] - \mathbf{E}[X]^2 \leq \mathbf{E}[X^2] \leq (w_i^{\max})^2 n!^2 \leq n^2 (n-1)^2 \kappa_i^2.
\end{align*}
Observe that $\bar{\kappa}_i$ is a random variable that is equal to $\frac{\sum_{j = 1}^{\lceil 4n^2(n-1)^2 / \epsilon^2 \rceil} X_j}{\lceil 4n^2(n-1)^2 / \epsilon^2 \rceil} , $
where $X_j$ are independent random variables with the same distribution as $X$. From this we obtain that $\mathbf{E}[\bar{\kappa}_i] = \mathbf{E}[X] = \kappa_i$.
The desired approximation guarantee then follows from Chebyshev's inequality,\footnote{One could also apply Hoeffding's inequality, but this will not result in an asymptotically better bound.} and completes the proof:
\begin{eqnarray*}
\mathbf{Pr}[|\bar{\kappa}_i - \kappa_i| \geq \epsilon \kappa_i] & \leq & \frac{\mathbf{Var}[\bar{\kappa}_i]}{\epsilon^2 \kappa_i^2} 
= \frac{\mathbf{Var}\left[\frac{1}{\lceil 4n^2(n-1)^2 / \epsilon^2 \rceil}\sum_{j = 1}^{\lceil 4n^2(n-1)^2 / \epsilon^2 \rceil} X_j\right]}{\epsilon^2 \kappa_i^2} \\
& = & \frac{\left(\frac{\mathbf{Var}[X]}{\lceil 4n^2(n-1)^2 / \epsilon^2 \rceil}\right)}{\epsilon^2 \kappa_i^2}
\leq  \frac{n^2(n-1)^2\kappa_i^2}{(4n^2(n-1)^2 / \epsilon^2) \cdot \epsilon^2 \kappa_i^2} \leq \frac14 .
\end{eqnarray*}
\end{proof}
\begin{corollary}
The algorithm that runs \textsc{MatchingGame-Sampler} and returns its output scaled down by $1/n!$, is an FPRAS for the Shapley value of a weighted matching game.
\end{corollary}

Observe that \textsc{MatchingGame-Sampler} is an FPRAS in the strong sense that its running time does not depend on the weights of the edges.
Due to the $\mathsf{\#P}$-completeness result stated in Theorem \ref{lemma:countinghard}, this FPRAS is the best one can hope for, and provides us with a complete answer to the precise complexity of this problem (based on our best judgment).

\paragraph{Acknowledgements.}
The authors thank Ross Kang for various helpful discussions.
				

\normalsize

\newpage
\section*{Appendix}

\paragraph{Graphs and Matchings Basics.}
Given an undirected graph $G = (N,E)$ (with vertex set $N$ and edge set $E$), a \emph{matching} of $G$ is a subset $M$ of $E$ such that $e \cap e' = \varnothing$ when $e,e' \in M$, $e \not=e'$. When discussing a particular matching $M$, we refer to the edges of a matching $M$ as \emph{matched edges}, and those outside $M$ as \emph{unmatched edges}. A \emph{matched graph} is a pair $(G,M)$ where $G$ is a graph and $M$ is a matching of $G$.
A \emph{maximum matching} of $G$ is a matching of maximum cardinality among the set of all matchings of $G$.

We call a vertex $i$ \emph{exposed} or \emph{unmatched} in $(G,M)$ when $i$ is not in any edge of $M$. Otherwise, we call $i$ \emph{matched}.
An \emph{alternating path} $P$ in $(G,M)$ is a path in $G$ where the edges of $P$ alternate between edges in $M$ and edges in $E \backslash M$.
An \emph{augmenting path} $P$ (with respect to a matching $M$) is an alternating path in $G$ of which the endpoints are both exposed vertices. An augmenting path thus has odd length, starts with an unmatched edge, and ends with an unmatched edge. 
The following lemma is fundamental to matching theory:
\begin{lemma}[Berge's lemma]\label{bergeslemma}
Let $G = (V,E)$ be a graph. A matching $M$ of $G$ maximum if and only if there is no augmenting path in $G$ with respect to $M$.
\end{lemma}

Suppose we have a matching $M$ for a graph $G$ that is not a maximum matching. Then by the above lemma, there is an augmenting path $P$. It can be seen that removing from $M$ the matched edges of $P$ and adding to $M$ the unmatched edges of $P$, gives us a bigger matching (i.e., a matching with one additional edge). We refer to this as the operation of \emph{augmenting $M$ along $P$}. Likewise, it is possible to augment a matching along an even-length alternating path with one exposed vertex and one matched vertex as endpoints. Augmenting along such a path does not increase the cardinality of the matching.

Observe that if $P$ is an alternating path that is not augmenting, then it still possible to augment the matching along $P$ iff one of the endpoints of $P$ is an exposed vertex.  
Edmonds' blossom algorithm~\citep{Edmo65a} is a polynomial time algorithm for finding a maximum weight matching in a graph.

Let $M_1$ and $M_2$ be two distinct maximum matchings for an unweighted graph $G = (V,E)$. Then $M_2$ can be obtained from $M_1$ by a sequence of augmentations along mutually disjoint even-length alternating paths and even-length alternating cycles. A rough sketch of a proof for this is as follows: We investigate the symmetric difference $D$ of $M_1$ and $M_2$, and conclude that $D$ must be a collection of disjoint even-length paths and even length cycles of which the edges alternate between edges in $M_1$ and edges in $M_2$. A cycle in $D$ must be an alternating cycle in $M_1$, and a path in $D$ must be an alternating path in $M_1$. After augmenting $M_1$ along such a cycle or path, we obtain a matching $M_3$ such that the symmetric difference between $M_3$ and $M_2$ is $D$ minus the cycle or path that we augmented. So by augmenting along all paths and cycles in $D$, we obtain $M_2$.

\newpage
\paragraph{Proof of Lemma~\ref{lemma:components}.}

We prove this for $k = 2$. For $k > 2$, the claim then holds by straightforward induction.

Therefore, let $N_1$ and $N_2$ be the vertex sets of the two connected components of $G$, and let $(N, v_{N_1}')$ and $(N, v_{N_2}')$ denote the matching game obtained by removing from the graph all edges among vertices in respectively $N_1$ and $N_2$. Note that $(V,v)$ is the sum of $v_{N_1}'$ and $v_{N_2}'$. 

By the additivity property of the Shapley value, it therefore holds for every player $i$ that $\varphi_i(v) = \varphi_i(v_{N_1}') + \varphi_i(v_{N_2}')$. It suffices to show that $\varphi_i(v_{N_1}') = \varphi_i(v_{N_2})$ for all $i \in N_1$ and that $\varphi_i(v_{N_2}') = \varphi_i(v_{N_2})$ for all $i \in N_2$. We do this by showing that $\varphi_i(V) = \varphi_i(V \cup \{j\})$, where $j$ is an arbitrary player from $U$. The claim for $k=2$ then follows by induction and symmetry. 

The fact that $\varphi_i(V) = \varphi_i(V \cup \{j\})$ holds, follows from the following derivation:
\begin{eqnarray*}
&   & \varphi_i(V) \\
& = & \frac{1}{|V|!} \sum_{S : S \subseteq V\backslash\{i\}} |S|!(|V|-|S|-1)!(v(S \cup \{i\})-v(S)) \\
& = & \frac{1}{(|V| + 1)!} \sum_{S : S \subseteq V \backslash \{i\}} (|S| + 1 + |V| - |S|)|S|!(|V| - |S| - 1)!(v(S \cup \{i\}) - v(S)) \\
& = & \frac{1}{|V \cup \{j\}|!} \sum_{S : S \subseteq V \backslash \{i\}} (|S|+1)!(|V| - |S| - 1)!(v(S \cup \{i\}) - v(S)) \\
&   &+ \frac{1}{|V \cup \{j\}|!} \sum_{S : S \subseteq V \backslash \{i\}} |S|!(|V|-|S|)!(v(S \cup \{i\}) - v(S)) \\
& = & \frac{1}{|V \cup \{j\}|!} \sum_{S : S \subseteq (V \cup \{j\}) \backslash \{i\}, j \in S} |S|!(|V \cup \{j\}| - |S| - 1)!(v(S \cup \{i\}) - v(S)) \\
&  &+ \frac{1}{|V \cup \{j\}|!} \sum_{S : S \subseteq (V \cup \{j\}) \backslash \{i\}, j \not\in S} |S|!(|V \cup \{j\}|-|S| - 1)!(v(S \cup \{i\}) - v(S)) \\
& = & \frac{1}{|V \cup \{j\}|!} \sum_{S : S \subseteq (V \cup \{j\}) \backslash \{i\}} |S|!(|V \cup \{j\}| - |S| - 1)!(v(S \cup \{i\}) - v(S)) \\
& = & \varphi_i(V \cup \{j\}).
\end{eqnarray*}
\qed

\newpage

\paragraph{Proof of Lemma~\ref{lemma:basic}.}

Let $\eta_i^s$ be the number of coalitions of size $s$ for which a vertex $i$ is pivotal.
\begin{eqnarray*}
\varphi_i(v) & = & \frac{1}{|V|!} \sum_{S : S \subseteq V\backslash\{i\}} |S|!(|V|-|S|-1)!(v(S \cup \{i\})-v(S)) \\
 & = & \frac{1}{|V|!} \sum_{s=1}^{|V|-1}\sum_{\substack{S : S \subseteq V\backslash\{i\}\\ |S|=s}}s!(|V|-s-1)!(v(S \cup \{i\})-v(S))\\
 & = & \frac{1}{|V|!} \sum_{s=1}^{|V|-1} s!(|V|-s-1)!\sum_{\substack{S : S \subseteq V\backslash\{i\})\\ |S|=s}}(v(S \cup \{i\})-v(S))\\
 & = & \frac{1}{|V|!} \sum_{s=1}^{|V|-1} s!(|V|-s-1)!\eta_i^s.
\end{eqnarray*}

Therefore, the problem of computing the Shapley value reduces to computing $\eta_i^s$ for all $s\in [0,\ldots, |V-1|]$. \qed

\newpage

\paragraph{Proof of Lemma \ref{lem:char} for even $i$.}
($\Rightarrow$) Let $M$ be a maximum matching for $G_i(S)$. $S$ is pivotal, so $M$ is not a perfect matching. We can assume though, that all vertices $\{y_0, \ldots, y_{i-1}\}$ are matched to each other in the matched graph $(G_i(S), M)$, because $G_i(\{y_0, \ldots, y_{i-1}\})$ is a linear graph with an even number of vertices, and is thus perfectly matchable. It follows that the exposed nodes of $(G_i(S), M)$ are all in $N$, and therefore the matching $M$ restricted to $N$ is a maximum matching for $G(S \backslash \{y_0, \ldots, y_{i-1}\}) = G(S \cap N)$ that is non-perfect. 

($\Leftarrow$) Let $M$ be a maximum (non-perfect) matching for $G(S \cap N)$ and let $y$ be an exposed vertex of $(G(S \cap N), M)$. Then $M' = M \cup \{\{y_j, y_{j + 1}\} : j \text{ even} \wedge j < i\}$ is a maximum matching for $G_i(S)$, by Berge's lemma (Lemma \ref{bergeslemma}), as it is clear that there is no augmenting path in $(G_i(S), M')$. Moreover, observe that in $(G_i(S), M')$ there is an even-length alternating path from $y$ to $y_{i-1}$. Therefore, there is in $(G_i, M')$ an augmenting path from $y$ to $y_i$, and it follows again by Berge's lemma that $S$ is pivotal. 

\paragraph{Proof of Lemma \ref{lem:char} for odd $i$.}

($\Rightarrow$) Let $M'$ be a maximum matching for $G_i(S)$. $S$ is pivotal, so in $(G_i(S),M')$ there is an even-length alternating path $P$ from an exposed node $y$ to $y_{i-1}$. Obtain the matching $M$ by augmenting $M'$ along $P$. $M$ is then a maximum matching for $G_i(S)$ in which $y_{i-1}$ is exposed. $G_i(\{y_0, \ldots, y_{i-1}\})$ is a linear graph and $M$ is maximum, so it follows that $y_{i-1}$ is the only exposed node in $(G_i(S),M)$ among $\{y_0, \ldots y_{i-1}\}$. Therefore $S \cap N$ must be matched to each other in $(G(S), M)$ (for otherwise, in $(G_i(S),M)$ there would be an augmenting path from $y_{i-1}$ to an exposed node of $S \cap N$, contradicting the fact that $M$ is a maximum matching for $G_i(S)$). It follows that $G(S \cap N)$ is perfectly matchable.

($\Leftarrow$) Let $M$ be a maximum perfect matching for $G(S \cap N)$. Let $M'$ be a maximum matching for $G_i(\{y_0, \ldots, y_{i-1}\})$ in which $y_{i-1}$ is the only exposed node. Then $M \cup M'$ is a matching for $G_i(S)$ in which $y_{i-1}$ is the only exposed node. $M \cup M'$ is clearly a maximum matching, and in $(G_i, M \cup M')$ the edge $\{y_{i-1}, y_i\}$ is exposed. So $S$ is pivotal. \qed

\end{document}